\documentclass[journal,twocolumn]{IEEEtran}
\IEEEoverridecommandlockouts
\usepackage{cite}
\usepackage{amsmath,amssymb,amsfonts,amsthm}
\usepackage{graphicx}

\usepackage{subcaption}
\usepackage{mathtools} % coloneqq
\usepackage[linesnumbered, ruled]{algorithm2e} % Algorithm writing
\usepackage{textcomp}
\usepackage{tabu} % Spacing out tables
\usepackage{xcolor}
\usepackage[font={small}]{caption}
\usepackage{multirow}

\newtheorem{theorem}{Theorem}

\newtheorem{lem}{Lemma}

\newtheorem{problem}{Problem}
\newtheorem{defn}{Definition}
\newtheorem{rem}{Remark}
\newtheorem{cor}{Corollary}

\usepackage{comment}

\def\BibTeX{{\rm B\kern-.05em{\sc i\kern-.025em b}\kern-.08em
	T\kern-.1667em\lower.7ex\hbox{E}\kern-.125emX}}

\renewcommand{\Re}{\mathbb{R}}

\newcommand{\panos}[1]{\textcolor{red}{#1}}
\newcommand{\josh}[1]{\textcolor{orange}{#1}}

\graphicspath{{./Figures/}}

\SetKwRepeat{Do}{do}{while}%

\begin{document}

\title{Data-Driven Covariance Steering Control Design\\}
\author{%
	Joshua Pilipovsky\textsuperscript{*} \thanks{\textsuperscript{*} J. Pilipovsky is a PhD student at the School of Aerospace Engineering, Georgia Institute of Technology, Atlanta, GA 30332-0150, USA. Email: jpilipovsky3@gatech.edu} 
	~~~~
	~~~~
	Panagiotis Tsiotras\textsuperscript{$\dagger$} \thanks{\textsuperscript{$\dagger$} P. Tsiotras is the David \& Lewis Chair and Professor at the School of Aerospace Engineering and the Institute for Robotics \& Intelligent Machines, Georgia Institute of Technology, Atlanta, GA 30332-0150, USA. Email: tsiotras@gatech.edu}
}%
\maketitle

\begin{abstract}
	This paper studies the problem of steering the distribution of a linear time-invariant system from an initial normal distribution to a terminal normal distribution under \textit{no} knowledge of the system dynamics.
	This data-driven control framework uses data collected from the input and the state and utilizes the seminal work by Willems et al. to construct a data-based parametrization of the mean and the covariance control problems.
	These problems are then solved to optimality as convex programs using standard techniques from the covariance control literature.
	We also discuss the equivalence of indirect and direct data-driven covariance steering designs, as well as a regularized version of the problem that provides a balance between the two. 
	We illustrate the proposed framework through a set of randomized trials on a double integrator system and show that the results match up almost exactly with the corresponding model-based method in the noiseless case.
	We then analyze the robustness properties of the data-free and data-driven covariance steering methods and demonstrate the trade-offs between performance and optimality among these methods in the presence of data corrupted with exogenous noise.
\end{abstract} 

\section{Introduction}
Recently, there has been an emergence of an increasing reliance on data-driven methods for solving complex problems in science and engineering. 
The field of artificial intelligence has demonstrated the ability to solve extremely difficult problems using input and output data using the machinery of neural networks and learning-based algorithms \cite{speechNN, imageClassNN, droneAcrobatics}.
One of the major fundamental flaws, however, of purely learning-based solutions is their lack of verifiability, that is, verifying that the networks will perform as expected given an input data stream.
Many works have begun looking at robustness properties to verify these neural networks both in the deterministic \cite{SMT1, exactVerification_fazlyab, VERIF_PAVONE, RPM_vincent, MILPNN1} and the probabilistic settings \cite{pilipovsky2022probabilistic, probVerification_fazlyab}.
The problem of analyzing the stability and robustness of a general learning-based solution is still intractable, however.

In the context of control theory, it is also often the case that we do not have prior knowledge of the system dynamics.
Ideally, we would like to use the data collected to perform control designs that have guaranteed performance and robustness properties, albeit in a learning/data-driven context.
To this end, another paradigm to perform control design is to use the input and output data streams to either estimate the model of the system, or directly perform controller synthesis using these data.
The former technique is referred to as an \textit{indirect} design, as it performs system identification (sysID) first, followed by controller synthesis, while the latter is reffered to as a \textit{direct} design, as it bypasses the sysID phase completely and directly generates control commands from input and output data.
Furthermore, each method may be classified into either a \textit{certainty-equivalence} (CE) or \textit{robust} approach, depending on whether uncertainties are taken into account.

The indirect or sysID approach has been long studied in the general setting \cite{Keesman2011, Ljung1998} with methods such as subspace identification with broader applications to filtering and state estimation \cite{Verhaegen2007}.
For optimization-based control, such as the linear quadratic regular (LQR), there are many works that use the indirect approach \cite{indirectLQR1, indirectLQR2, indirectLQR3, indirectLQR4, indiretLQR5}.
Similarly, the direct approach to data-driven LQR has been approached using behavioral methods \cite{DataDriven_dePersis_1}, gradient-based methods \cite{directLQR_gradient1, directLQR_gradient2}, and Riccati-based methods \cite{directLQR_Riccati}.
More recently, this problem has been solved using concepts from behavioral systems theory and subspace methods \cite{behavioral_system_theory} using \textit{Willems' Fundamental Lemma}, which characterizes the trajectory of an LTI system through the range space of the input/output data matrix \cite{WFL}.
This gives rise to a parametrization of the feedback gains as a linear combination of the collected data and allows to perform direct data-driven control design through semi-definite programming (SDP) \cite{DataDriven_dePersis_2}.
To this end, the authors in \cite{DataDriven_CE, DataDriven_regularization} were able to come up with a CE and regularized design, thus bridging the gap between the indirect and direct approaches.

All these optimization-based approaches to direct data-driven LQR design are done in the deterministic setting, assuming a single realization of the state trajectories from $x_0$ to $x_f$. 
To the best of our knowledge, this paper is the first work that looks at the problem where the state is a stochastic process instead of a deterministic trajectory.
Since it is still unclear how to perform data-driven designs in the context of process noise, we limit our analysis to uncertainties in the initial and final states, leaving the case of process noise for future investigation.
Specifically, we assume that the boundary values of the state must follow a normal distribution, and the objective is to steer the entire \textit{distribution} of states from an initial to a final one.
This problem is referred to in the literature as \textit{covariance steering} (CS) and has been extensively studied in the past couple of years in the Gaussian case \cite{Max1, EB1, Halder, kazu_PP, exact_CS}, non-Gaussian case \cite{nonGaussianCS}, nonlinear case \cite{nonlinearCS}, as well as in the presence of chance-constraints on the state and control \cite{JP1, EB2, JoshJack}.
The CS problem in the unconstrained case has an elegant separation property between the mean and covariance dynamics, and both problems turn into convex programs \cite{Max1}.

In this work, we present a data-driven covariance control design that steers the mean dynamics through an indirect sysID approach, and steers the covariance dynamics through a direct approach, using the techniques outlined in \cite{exact_CS} and \cite{DataDriven_dePersis_2}.
For the direct approach of covariance steering the certainty-equivalence approach is used, which adds an orthogonality constraint to the resulting optimization problem.
We also analyze the regularized approach, which is a hybrid of the two.
We finally study the robustness of these approaches through a set of randomized simulations with noisy data and compare the various covariance steering frameworks.

The paper is organized as follows. Section~\ref{sec:PS} introduces the data-driven CS problem. 
Section~\ref{sec:PS_reformulation} defines the control policy and reformulates the problem into the mean and the covariance subproblems. 
Section~\ref{sec:DataDrivenDesign} reviews the Fundamental Lemma and derives the data-driven convex programs for both the mean and the covariance steering problems.
Section~\ref{sec:CE_reg} discusses the CE and regularization approaches to the direct covariance steering design.
Section~\ref{sec:sims} presents our numerical case studies, and, lastly, Section~\ref{sec:conclusion} concludes the paper.

%---------------------------------------------------%
\section{Problem Statement}~\label{sec:PS}

We consider the following discrete-time \textit{deterministic} time-invariant system
\begin{equation}~\label{eq:1}
	x_{k+1} = A x_k + B u_k,
\end{equation}
where $x\in\mathbb{R}^n,u\in\mathbb{R}^m$, with time steps $k = 0,\ldots,N-1$, where $N$ representing the finite horizon. 
The system matrices $A$ and $B$ are assumed to be unknown.
The uncertainty in the system resides in the initial state $x_0$, which is a random $n$-dimensional vector drawn from the normal distribution
\begin{equation}~\label{eq:2}
	x_0 \sim \mathcal{N}(\mu_i, \Sigma_i),
\end{equation}
where $\mu_0\in\mathbb{R}^n$ is the initial state mean and $\Sigma_0\in\mathbb{R}^{n\times n} \succ 0$ is the initial state covariance. 
Thus, although the system dynamics \eqref{eq:1} is deterministic, the uncertainty in the initial state makes the state evolution $\{x_k\}_{k=1}^{N}$ a random process. 
The objective is to steer the trajectories of (\ref{eq:1}) from the initial distribution (\ref{eq:2}) to the terminal distribution
\begin{equation}~\label{eq:3}
	x_N = x_f \sim \mathcal{N}(\mu_f, \Sigma_f),
\end{equation}
where $\mu_f\in\mathbb{R}^n$ and $\Sigma_f \succ 0$ are the desired state mean and covariance at time $N$, respectively. 
The cost function to be minimized is 
\begin{equation}~\label{eq:4}
	J(u_0,\ldots,u_{N-1}) \coloneqq \mathbb{E}\bigg[\sum_{k=0}^{N-1}x_k^\intercal Q_k x_k + u_k^\intercal R_k u_k\bigg],
\end{equation}
where $Q_k\succeq 0$ and $R_k \succ 0$ for all $k = 0,\ldots,N-1$. 

\begin{comment}
	\textit{Remark 1}: We assume that the system (\ref{eq:1}) is controllable, that is, for any $x_0,x_f\in\mathbb{R}^n$, and no noise ($w_k\equiv 0,~ k = 0,\ldots,N-1$), there exists a sequence of control inputs $\{u_k\}_{k=0}^{N-1}$ that steer the system from $x_0$ to $x_f$.
\end{comment}

\begin{problem}~\label{problem:1}
	Given the unknown linear system (\ref{eq:1}), find the optimal control sequence $\{u_k\}_{k=0}^{N-1}$ that minimizes the objective function (\ref{eq:4}), subject to the initial state (\ref{eq:2}) and terminal state (\ref{eq:3}).
\end{problem}

\section{Problem Reformulation}~\label{sec:PS_reformulation}

Borrowing from the work in \cite{exact_CS}, we adopt the control policy
\begin{equation}~\label{eq:controlLaw}
	u_k = K_k(x_k - \mu_k) + v_k,
\end{equation}
where $K_k\in\Re^{m\times n}$ are the feedback gains that control the covariance of the state, and $v_k\in\Re^{m}$ is the feedforward term that controls the mean of the state.
Under this control law, it is possible to re-write Problem~\ref{problem:1} as a convex program, which can be solved to optimality using off-the-shelf solvers.
Since the state distribution remains Gaussian at all time steps, and since a normal distribution is completely characterized by its first two moments, we decompose the system dynamics \eqref{eq:1} into the mean dynamics and covariance dynamics.
Plugging in the control law \eqref{eq:controlLaw} into the dynamics \eqref{eq:1} yields the decoupled dynamics
\begin{subequations}~\label{eq:Dynamics}
	\begin{align}
		\mu_{k+1} &= A\mu_k + B v_k, \label{eq:meanDynamics} \\
		\Sigma_{k+1} &= (A + B K_k)\Sigma_k(A + B K_k)^\intercal. \label{eq:covDynamics}
	\end{align}
\end{subequations}
In the sequel, and similar to \cite{exact_CS_2}, we treat the moments of the intermediate states $\{\Sigma_k, \mu_k\}_{k=1}^{N-1}$ in the steering horizon as decision variables in the resulting optimization problem.

Similar to the dynamics, the cost function can be decoupled and written in terms of the first two moments as follows
\begin{subequations}
	\begin{align}
		J &= J_{\mu}(\mu_k, v_k) + J_{\Sigma}(\Sigma_k, K_k), \\
		J_{\mu} &:= \sum_{k=0}^{N-1}\left(\mu_k^\intercal Q_k \mu_k + v_k^\intercal R_k v_k\right), \label{eq:meanCost} \\
		J_{\Sigma} &:= \sum_{k=0}^{N-1}\Big(\textrm{tr}(Q_k\Sigma_k) + \textrm{tr}(R_kK_k\Sigma_k K_k^\intercal) \Big). \label{eq:covCost}
	\end{align}
\end{subequations}
Lastly, the two boundary conditions are written as 
\begin{subequations}
	\begin{align}
		&\mu_0 = \mu_i, \quad \mu_N = \mu_f, \label{eq:meanBCs} \\
		&\Sigma_0 = \Sigma_i, \quad \Sigma_N = \Sigma_f. \label{eq:covBCs}
	\end{align}
\end{subequations}
Problem~\ref{problem:1} is now recast as the following two sub-problems.

\begin{problem}~\label{problem:2a}
	Given the mean dynamics \eqref{eq:meanDynamics}, find the optimal mean trajectory $\{\mu_k\}_{k=1}^{N-1}$ and feedforward control $\{v_k\}_{k=0}^{N-1}$ that minimize the mean cost \eqref{eq:meanCost} subject to the boundary conditions \eqref{eq:meanBCs}.
\end{problem} 

\begin{problem}~\label{problem:2b}
	Given the covariance dynamics \eqref{eq:covDynamics}, find the optimal covariance trajectory $\{\Sigma_k\}_{k=1}^{N-1}$ and feedback gains $\{K_k\}_{k=0}^{N-1}$ that minimize the covariance cost \eqref{eq:covCost} subject to the boundary conditions \eqref{eq:covBCs}.
\end{problem} 

\begin{rem}~\label{rem:1}
	Both the mean and covariance steering problems rely on the system matrices $A$ and $B$ through the system dynamics \eqref{eq:Dynamics}.
	Thus, the problems, as stated above, are not yet amenable to a data-driven solution.
\end{rem} 

\begin{rem}~\label{rem:2}
	Problem~\ref{problem:2a} is a standard quadratic program with linear constraints that can be solved analytically given knowledge of the system matrices \cite{Max1}.
	As such, we will perform an indirect design by first estimating the $A$ and $B$ matrices to solve this problem in a data-driven fashion.
	Problem~\ref{problem:2b}, however, is a non-linear program due to the cost term $\mathrm{tr}(R K_k \Sigma_k K_k^\intercal)$ and the covariance dynamics.
\end{rem}

In the following section, we review the main concepts from behavioral systems theory that will allow us to parametrize the decision variables in Problems~\ref{problem:2a} and \ref{problem:2b} in terms of input and output data streams.

\section{Data-Driven Parameterization}~\label{sec:DataDrivenDesign}

In order to incorporate data into the problem formulation, we use the concept of persistence of excitation, along with Willems' Fundamental Lemma \cite{WFL} to parametrize the feedforward and feedback gains of the control policy.
First, recall the following definitions.

\begin{defn}~\label{def:1}
	Given a signal $\{z_k\}_{k=0}^{T-1}$ where $z\in\Re^{\sigma}$, we denote its Hankel matrix by
	\begin{equation}
		Z_{i,\ell, j} := 
		\begin{bmatrix}
			z_i & z_{i+1} & \ldots & z_{i+j-1} \\
			z_{i+1} & z_{i+2} & \ldots & z_{i + j} \\
			\vdots & \vdots & \ddots & \vdots \\
			z_{i+\ell-1} & z_{i+\ell} & \ldots & z_{i+\ell+j-2}
		\end{bmatrix} \in \Re^{\sigma\ell \times j},
	\end{equation}
	where $i\in\mathbb{Z}$ and $\ell,j\in\mathbb{N}$.
	For shorthand notation, if $\ell=1$, we denote the Hankel matrix by
	\begin{equation}
		Z_{i,1,j} \equiv Z_{i,j} := [z_i \ z_{i+1} \ \ldots \ z_{i+j-1}].
	\end{equation}
\end{defn}

\begin{defn}~\label{def:2}
	The signal $\{z_k\}_{k=0}^{T-1}: [0,T-1]\cap\mathbb{Z}\rightarrow\Re^{\sigma}$ is \textit{persistently exciting} of order $\ell$ if the matrix $Z_{0,\ell,j}$ with $j = T - \ell + 1$ has rank $\sigma\ell$.
\end{defn}

\begin{cor}~\label{rem:3}
	In order for a signal to  be persistently exciting of order $\ell$, it must be sufficiently long, i.e., it must hold that $T \geq (\sigma + 1)\ell - 1$.
\end{cor} 

Suppose we carry out an experiment of duration $T\in\mathbb{N}$ where we collect input and state data $\{u_k\}_{k=0}^{T-1}$ and $\{x_k\}_{k=0}^{T}$, respectively.
Let the corresponding Hankel matrices for the input sequence, state sequence, and shifted state sequence (with $\ell = 1$) be 
\begin{subequations}
	\begin{align}
		U_{0,T} &:= [u_0 \ u_1 \ \ldots \ u_{T-1}], \\
		X_{0,T} &:= [x_0 \ x_1 \ \ldots \ x_{T-1}], \\
		X_{1,T} &:= [x_1 \ x_2 \ \ldots \ x_{T}].
	\end{align}
\end{subequations}

The next result characterizes the rank of the stacked Hankel matrices of the input and output data, and is central to the method used to formulate a tractable data-driven covariance steering problem.

\begin{lem}[Willems' Fundamental Lemma\cite{WFL}] \label{lemma1}
	Suppose that system \eqref{eq:1} is controllable.
	If the input signal $\{u_k\}_{k=0}^{T-1}$ is persistently exciting of order $n + 1$, then
	\begin{equation}
		\mathrm{rank}
		\begin{bmatrix}
			U_{0,T} \\
			X_{0,T}
		\end{bmatrix}
		= n + m.
	\end{equation}
\end{lem}

\begin{rem}~\label{rem:4}
	In order to ensure that the input $u_k\in\Re^{m}$ is persistently exciting of order $n + 1$ to satisfy Willems' Fundamental Lemma, it is sufficient to check that $T \geq (m + 1)n + m$. 
	In practice, this can always be achieved in real-time during data collection.
\end{rem} 

Lemma~\ref{lemma1} implies that \textit{any} arbitrary input-state sequence can be expressed as a linear combination of the collected input-state data. 
Furthermore, this can be extended \cite{DataDriven_dePersis_1} to parameterizing any arbitrary feedback interconnection as well.
In the following section, based on the work in \cite{DataDriven_dePersis_2}, we parameterize the feedback gains in terms of the input-state data and reformulate the covariance steering problem as a semi-definite program (SDP). 

\subsection{Direct Data-Driven Covariance Steering}~\label{sec:directCS}
Assuming the signal $\{u_k\}_{k=0}^{T-1}$ is persistently exciting of order $n + 1$, we can express the feedback gains as follows
\begin{equation}~\label{eq:WFL_gains}
	\begin{bmatrix}
		K_k \\
		I_n
	\end{bmatrix}
	= 
	\begin{bmatrix}
		U_{0,T} \\
		X_{0,T}
	\end{bmatrix}
	G_k,
\end{equation}
where $G_k \in \Re^{T\times n}$ are newly defined decision variables that provide the link between the feedback gains and the input-state data.
Furthermore, we can re-write the covariance dynamics \eqref{eq:covDynamics} as
\begin{align}  \label{eq:covDynamics_Gvars}
	\Sigma_{k+1} &= [B \ \ A]
	\begin{bmatrix}
		K_k \\
		I_n
	\end{bmatrix}
	\Sigma_k 
	\begin{bmatrix}
		K_k \\
		I_n
	\end{bmatrix}^\intercal
	[B \ \ A]^\intercal \nonumber \\
	&= X_{1,T} G_k \Sigma_k G_k^\intercal X_{1,T}^\intercal,
\end{align}
where we use the fact that $X_{1,T} = A X_{0,T} + B U_{0,T}$.
Similarly, the covariance cost \eqref{eq:covCost} can be re-written as 
\begin{equation}~\label{eq:covCost_Gvars}
	J_{\Sigma, k} = \mathrm{tr}(Q_k\Sigma_k) + \mathrm{tr}(R_k U_{0,T} G_k \Sigma_k G_k^\intercal U_{0,T}^\intercal).
\end{equation}
To remedy the nonlinearity $G_k\Sigma_k G_k^\intercal$ in the covariance dynamics and the cost, define the new decision variables $S_k := G_k \Sigma_k\in\Re^{T\times n}$, which yields the covariance dynamics 
\begin{equation}~\label{eq:covDynamics_Svars}
	\Sigma_{k+1} = X_{1,T} S_k \Sigma_k^{-1} S_k X_{1,T}^\intercal,
\end{equation}
and the covariance cost
\begin{equation}~\label{eq:covCost_Svars}
	J_{\Sigma, k} = \mathrm{tr}(Q_k \Sigma_k) + \mathrm{tr}(R_k U_{0,T} S_k \Sigma_k^{-1} S_k^\intercal U_{0,T}^\intercal).
\end{equation}
This problem is still non-convex due to the nonlinear term $S_k\Sigma_k^{-1}S_k^\intercal$.
To this end, let us relax the covariance dynamics by defining a new decision variable $Y_k \succeq S_k \Sigma_k^{-1}S_k^\intercal$, which yields the relaxed optimization problem
\begin{subequations}~\label{eq:convexProblem}
	\begin{align}
		&\hspace*{-4mm} \min_{\Sigma_k, S_k Y_k} \bar{J}_{\Sigma} = \sum_{k=0}^{N-1}\left(\mathrm{tr}(Q_k\Sigma_k) + \mathrm{tr}(R_kU_{0,T}Y_kU_{0,T}^\intercal)\right), \label{eq:convexProblem_cost}
	\end{align}
	such that, for all $k = 0,\ldots, N - 1$,
	\begin{align}
		&C_k := S_k \Sigma_k^{-1} S_k^\intercal - Y_k \preceq 0, \label{eq:convexProblem_ineqConstraint} \\
		&G_k^{(1)} := X_{1,T} Y_k X_{1,T}^\intercal - \Sigma_{k+1} = 0, \label{eq:convexProblem_eqConstraint1} \\
		&G_k^{(2)} := \Sigma_k - X_{0,T} S_k = 0, \label{eq:convexProblem_eqConstraint2}
	\end{align}
\end{subequations}
with the boundary conditions \eqref{eq:covBCs}.
The last equality constraint \eqref{eq:convexProblem_eqConstraint2} comes from the second block in \eqref{eq:WFL_gains} by multiplying $\Sigma_k$ on the right.
The relaxed problem \eqref{eq:convexProblem} is convex, since the constraint \eqref{eq:convexProblem_ineqConstraint} can be written using the Schur complement as the linear matrix inequaltiy (LMI)
\begin{equation}
	\begin{bmatrix}
		\Sigma_k & S_k^\intercal \\
		S_k & Y_k
	\end{bmatrix} \succeq 0.
\end{equation}
The equality constraint \eqref{eq:convexProblem_eqConstraint1} and cost \eqref{eq:convexProblem_cost}, on the other hand, are simply linear in all the decision variables, and hence convex.

\begin{comment}
	\begin{rem}~\label{rem:5}
		The convex reformulation proposed in this work is very similar to that of [ref - George], which proved the relaxation is lossless.
		However, in the current work, this is indeed a relaxation due to the equality constraint \eqref{eq:convexProblem_eqConstraint1}, which is only a function of the decision variables $Y_k$, while in [ref - George], $G_k$ is a function of both $Y_k$ and $S_k$, yielding the relaxed problem lossless through strong duality.
		The implication of this relaxation is that the terminal covariance resulting from this data-driven method will be smaller than the desired one, that is, $\Sigma_N \preceq \Sigma_f$.
		\panos{is this needed?} \josh{probably not}
	\end{rem} 
\end{comment}

%
\subsection{Indirect Data-Driven Mean Steering}~\label{sec:meanSteering}

Given the mean dynamics \eqref{eq:meanDynamics} in terms of the open-loop control $v_k$, Lemma~\ref{lemma1} also provides a system identification type of result using the following theorem.
\begin{theorem}~\label{theorem:1}
	Suppose the signal $u_k$ is persistently exciting of order $n + 1$.
	Then, the system \eqref{eq:meanDynamics} has the following equivalent representation
	\begin{equation}  \label{eq:mut+1}
		\mu_{k+1} = X_{1,T}
		\begin{bmatrix}
			U_{0,T} \\
			X_{0,T}
		\end{bmatrix}^{\dagger}
		\begin{bmatrix}
			v_k \\
			\mu_k
		\end{bmatrix}.
	\end{equation}
\end{theorem}

\begin{proof}
	See Appendix~A.
\end{proof}

\begin{rem}
	Theorem~\ref{theorem:1} gives a data-based open-loop representation of a (noise-less) linear system.
	One may equivalently interpret equation~\eqref{eq:mut+1}
	as the solution to the least-squares problem
	\begin{equation}
		\min_{B,A} \left\|X_{1,T} - [B \ \ A]
		\begin{bmatrix}
			U_{0,T} \\
			X_{0,T}
		\end{bmatrix}
		\right\|_{F},
	\end{equation}
	where $\|\cdot\|_{F}$ is the Frobenius norm.
	Thus, equation \eqref{eq:mut+1} provides
	the solution for the system matrices that best approximates the system dynamics.
\end{rem}

Using Theorem~\ref{theorem:1}, we can express the mean steering problem as the following convex problem
\begin{subequations}~\label{eq:meanProblem}
	\begin{align}
		&\qquad\qquad\min_{\mu_k, v_k} J_{\mu} = \sum_{k=0}^{N-1} (\mu_k^\intercal Q_k \mu_k + v_k^\intercal R_k v_k), \label{eq:meanProblem_cost}
	\end{align}
	such that, for all $k = 0,\ldots, N - 1,$
	\begin{align}
		&\qquad\qquad H_k^{(1)} := F_{\mu} \mu_k + F_{v} v_k - \mu_{k+1} = 0, \label{eq:meanProblem_eqConstraint1}
	\end{align}
\end{subequations}
with the boundary conditions \eqref{eq:meanBCs}, where $F_{\mu}\in\Re^{n\times n}$ and $F_{v}\in\Re^{n\times m}$ result from the partition of
\begin{equation}
	F := X_{1,T}
	\begin{bmatrix}
		U_{0,T} \\
		X_{0,T}
	\end{bmatrix}^{\dagger} = 
	\begin{bmatrix}
		F_{v} \\
		F_{\mu}
	\end{bmatrix} \in \Re^{n\times (m + n)}.
\end{equation}

\section{Certainty Equivalence and Regularized Data-Driven Methods}~\label{sec:CE_reg}

In this section, we establish the link between the direct CS design in \ref{sec:directCS} and the indirect design, as well as briefly outline a regularized design based on \cite{DataDriven_regularization} that trade-offs the two frameworks.
\subsection{Certainty-Equivalence Design}~\label{subsec:CE}
For notational simplicity, let 
\begin{equation*}
	W_0 := 
	\begin{bmatrix}
		U_{0,T} \\
		X_{0,T}
	\end{bmatrix}.
\end{equation*}
In the direct data-driven covariance control design, the set of optimal solutions $G_k^*$ to \eqref{eq:convexProblem} coincides with the set of solutions to \eqref{eq:WFL_gains}, that is,
\begin{equation}
	\{G_k^* : (\Sigma_k^*, S_k^*, Y_k^*) \in \mathrm{argmin}\eqref{eq:convexProblem}\} = W_0^{\dagger}
	\begin{bmatrix}
		K_k^* \\
		I_n
	\end{bmatrix} + G_{\mathrm{hom}},
\end{equation}
where $G_{\mathrm{hom}}$ is any matrix in the null space of $W_0$. 
Let $\Pi_{W_0} := I_{T} - W_0^\dagger W_0$ be the orthogonal projection on the nullspace of $W_0$.
It has been shown in \cite{DataDriven_CE} that by introducing the extra \textit{orthogonality} constraint $\Pi_{W_0} G_k = 0$, for all $k = 1,\ldots, N$, results in an optimization problem \eqref{eq:convexProblem} that is \textit{equivalent} to the corresponding indirect design, which amounts to first performing system identification then control design on the approximate system, as outlined in Section~\ref{sec:meanSteering}.
Since $S_k$ is a decision variable in \eqref{eq:convexProblem}, this amount to adding the additional equality constraints $\Pi_{W_0} S_k = 0$ to \eqref{eq:convexProblem}.
In the context of covariance steering, the indirect design is equivalent to the \textit{bi-level} program
\begin{subequations}~\label{eq:biLevelProblem}
	\begin{align}
		&\min_{\Sigma_k, U_k Y_k} \bar{J}_{\Sigma} = \sum_{k=0}^{N-1}  \mathrm{tr}(Q_k\Sigma_k) + \mathrm{tr}(R_kY_k) , \label{eq:biLevelProblem_cost} 
	\end{align}
	such that, for all $k = 0,\ldots, N - 1$,
	\begin{align}
		&C_k^{(1)} := 
		\begin{bmatrix}
			\Sigma_k & U_k^\intercal \\
			U_k & Y_k
		\end{bmatrix} \succeq 0, \label{eq:biLevelProblem_ineqConstrain1} \\
		&C_k^{(2)} := 
		\begin{bmatrix}
			\Sigma_k & U_k^\intercal \hat{B}^\intercal \\
			\hat{B} U_k & \hat{\Gamma}_k
		\end{bmatrix} \succeq 0, \label{eq:biLevelProblem_ineqConstraint2} \\
		&[\hat{B} \ \ \hat{A}] = \mathrm{argmin}_{B, A} \|X_{1} - [B \ \ A] W_0\|_{F},
	\end{align}
\end{subequations}
where $\Gamma_k := \Sigma_{k+1} - \hat{A}\Sigma_k\hat{A}^\intercal - \hat{A}U_k^\intercal \hat{B}^\intercal - \hat{B}U_k\hat{A}^\intercal$.
See Appendix~B for details on this derivation, which is based on the work in \cite{exact_CS}.
\subsection{Regularized Design}

By adding the constraint $\Pi_{W_0}S_k = 0$ to the objective function, we arrive at a regularized direct data-driven covariance steering formulation.
Letting $\lambda\geq 0$ be a tunable hyperparameter that balances indirect with direct designs, the regularized problem becomes
\begin{subequations}~\label{eq:regularizedProblem}
	\begin{align}
		&\min_{\Sigma_k, S_k Y_k} \bar{J}_{\Sigma} = \sum_{k=0}^{N-1}\mathrm{tr}(Q_k\Sigma_k) + \mathrm{tr}(R_kU_{0,T}Y_kU_{0,T}^\intercal) \nonumber \\
		&\qquad\qquad\qquad\qquad\qquad\qquad\qquad + \lambda \|\Pi_{W_0} S_k\|, \label{eq:regularizedProblem_cost}
	\end{align}
	such that, for all $k = 0,\ldots, N - 1$,
	\begin{align}
		&C_k := S_k \Sigma_k^{-1} S_k^\intercal - Y_k \preceq 0, \label{eq:regularizedProblem_ineqConstraint} \\
		&G_k^{(1)} := X_{1,T} Y_k X_{1,T}^\intercal - \Sigma_{k+1} = 0, \label{eq:regularizedProblem_eqConstraint1} \\
		&G_k^{(2)} := \Sigma_k - X_{0,T} S_k = 0. \label{eq:regularizedProblem_eqConstraint2}
	\end{align}
\end{subequations}
It can be shown \cite{DataDriven_regularization} that for $\lambda$ sufficiently large, the regularized design \eqref{eq:regularizedProblem} coincides with the certainty-equivalence design.

\section{Numerical Example}~\label{sec:sims}
To illustrate the proposed data-driven method, we run a set of 100 trials on the double integrator system
\begin{equation}
	A = 
	\begin{bmatrix}
		1 & 1 \\
		0 & 1
	\end{bmatrix}, \quad 
	B = 
	\begin{bmatrix}
		0 \\
		1
	\end{bmatrix},
\end{equation}
with initial distribution $\mu_0 = [20 \ -2]^\intercal, \Sigma_0 = \mathrm{diag}(1, 0.5)$ and terminal distribution $\mu_f = [0 \ 8]^\intercal, \Sigma_f = 0.5I_2$.
The state and control cost weights are $Q_k = 0.01I_2$ and $R_k = 1$, for all $k = 0,1,\ldots, N$, respectively.
We pick a control horizon $N = 10$ and data collection horizon $T = 15$ to ensure that $T \geq (m + 1)n + m = 5$.  
The data is generated for every trial by randomly sampling the initial state and input over the collection horizon $T$ from a standard normal distribution.
The following set of simulations was run on a 32 GB Intel i7-10850H @ 2.60 GHz computer.

\begin{figure}[!htb]
	\centering
	\hspace*{-0.5cm}
	\includegraphics[scale=0.29]{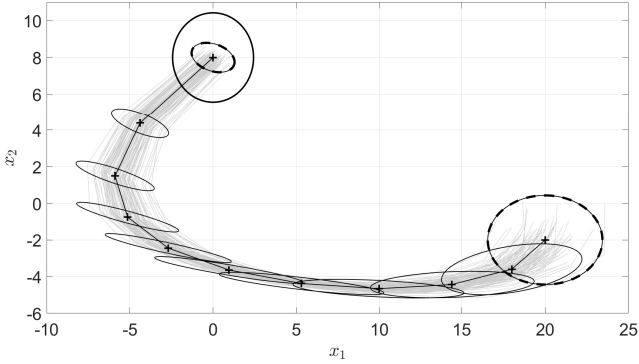}
	\caption{Monte Carlo trajectories of data-driven solution.}
	\label{fig:1}
\end{figure}

Figure~\ref{fig:1} shows the trajectories of the system using the data-driven framework on a set of 100 Monte Carlo runs.
As mentioned, the terminal covariance is indeed less than the desired one, as denoted by the solid black line.
Nevertheless, the control law successfully steers the system between the two distributions with no knowledge of the system matrices.

It is also fruitful to compare the gains and feedforward control to that of the \textit{model-based} covariance steering solution, as outlined in Appendix~B. 
Figure~\ref{fig:2} shows the difference between the two solutions with the corresponding mean and 3$\sigma$ errors over the set of trial runs.
\begin{figure}
	\centering
	\hspace*{-0.2cm}
	\begin{subfigure}[b]{0.55\textwidth}
		\includegraphics[width=0.95\linewidth]{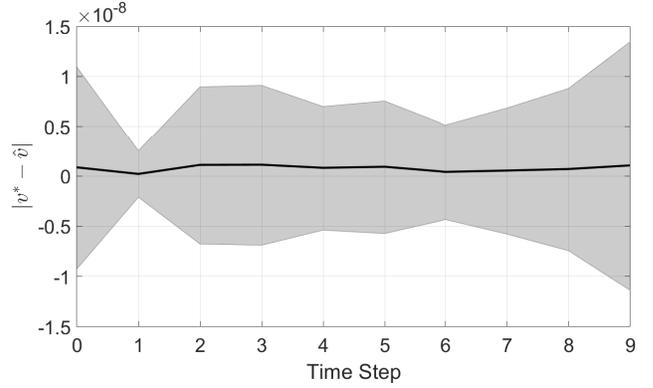}
		\caption{Absolute difference in feedforward control.}
		\label{fig:vdeltas} 
	\end{subfigure}
	\hspace*{-0.2cm}
	\begin{subfigure}[b]{0.55\textwidth}
		\includegraphics[width=0.95\linewidth]{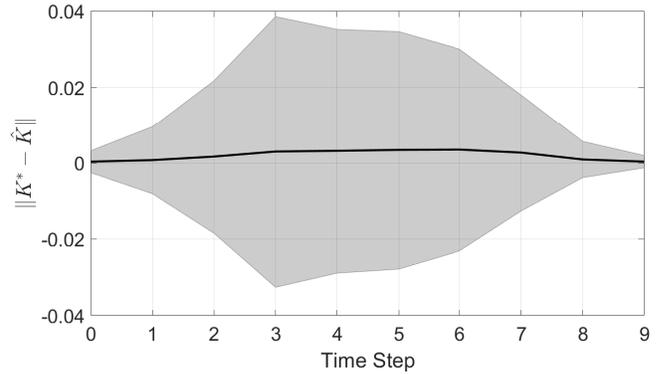}
		\caption{Two-norm difference in feedback control.}
		\label{fig:Kdeltas}
	\end{subfigure}
	\caption{Comparison of differences between model-based and data-driven covariance steering solutions in the noiseless case ($\beta = 0$) with mean errors (solid line) and 3$\sigma$ shaded error bars.}
	\label{fig:2}
\end{figure}
Figure~\ref{fig:2} shows that the two solutions are nearly exact, with the feedforward control having an error of within $10^{-8}$ and the feedback control within $10^{-2}$. 

Next, we compare the robustness properties of the various data-driven control designs as outlined in Section~\ref{sec:CE_reg}.
To this end, we add noise into the data collection and simulations but keep the designs as if there is no noise in the system.
The question then becomes how these methods will perform when the data is corrupted by noise.
For the data collection, we add an extra term $D w_k$ into the dynamics, where $w_k \sim \mathcal{N}(0,I_{n})$ and $D = \beta I_{n}$, with $\beta>0$ a tunable parameter for the noise intensity.
Figure~\ref{fig:FFcontrol_robustness} shows the effect of increasing levels of noise on the accuracy of the feedforward control data-driven solution. 
Even for moderately large noise levels ($\beta = 10^{-2}$), the 3$\sigma$ variance of the error is within 0.04.
Thus, indirect mean steering design is a fruitful avenue for designing robust nominal controllers against noise.

\begin{figure}[!htb]
	\centering
	\hspace*{-0.5cm}
	\includegraphics[scale=0.31]{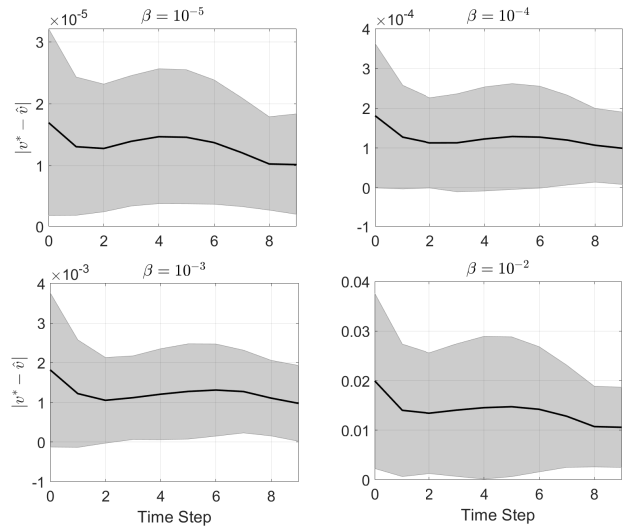}
	\caption{Absolute value error in data-driven feedforward control for different levels of exogenous noise.}
	\label{fig:FFcontrol_robustness}
\end{figure}

For the feedback design, we compare the direct solution as outlined in Section~\ref{sec:directCS} with the certainty-equivalence and regularized designs.
To this end, we fix the noise intensity to be $\beta = 10^{-3}$ and choose the regularization hyperparameter as $\lambda = 2\times 10^{-2}$.
Figure~\ref{fig:feedbackGain_robustness} shows the mean and 3$\sigma$ variance of the normed error of the corresponding data-driven methods when corrupted with noise.
Interestingly, the certainty-equivalence (i.e., indirect) approach is the \textit{most} robust among the three, achieving an error within $10^{-2}$, while the direct approach is the least robust.
\begin{figure}[!htb]
	\centering
	\hspace*{-0.5cm}
	\includegraphics[scale=0.35]{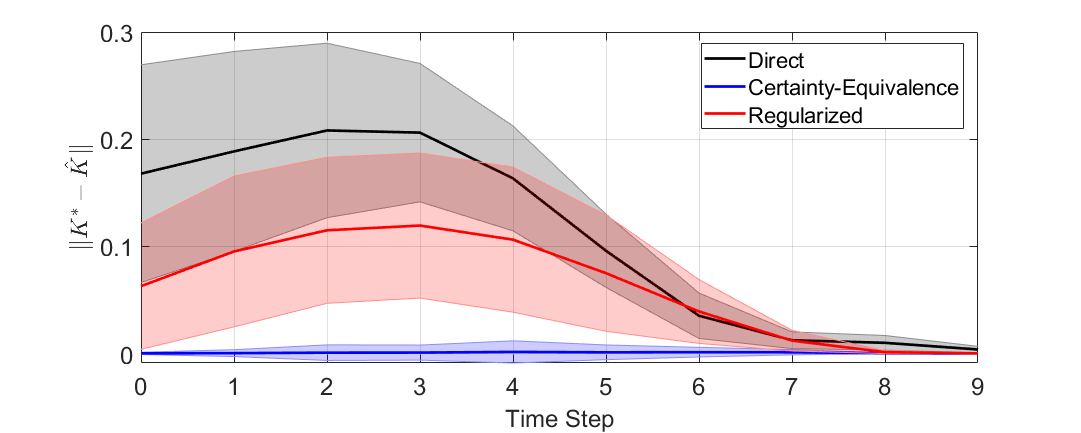}
	\caption{Norm error in data-driven feedback control for direct, regularized, and certainty-equivalence methods for a noise level $\beta=10^{-3}$ and regularization hyperparameter $\lambda=2\times 10^{-2}$.}
	\label{fig:feedbackGain_robustness}
\end{figure}
Table~\ref{table:covSteering_optimality} shows the mean and variance of the empirical error for the $k$th trial
\begin{equation}
	\mathcal{E}_k := \frac{|J^* - \hat{J}_k|}{J^*},
\end{equation}
between the optimal model-based cost and the data-driven costs.
\renewcommand{\arraystretch}{1.5} % Default value: 1
\begin{table}[!htb]
	\centering
	\caption{Mean cost error among different data-driven covariance steering approaches for corrupted data.}
	\begin{tabular}{llll}
		\hline
		$\times 10^4$ & Direct & Regularized & CE \\
		\hline
		$\mathbb{E}[\mathcal{E}_k]$ & 3.9886 & 4.3614 & 12.0494 \\
		$\sigma^2[\mathcal{E}_k]$ & 0.0009 & 0.0013 & 0.0261 \\
		\hline
	\end{tabular}
	\label{table:covSteering_optimality}
\end{table}
We see that although the direct data-driven solution is the \textit{least} robust, it is also the \textit{most} optimal in the sense that it achieves the closest cost with the true, model-based solution.
On the other hand, the CE approach has the best robustness but worst performance.
Thus, the regularized form of the problem is a useful design for balancing robustness with optimality by varying $\lambda$.

%---------------------------------------------------%

\section{Conclusion}~\label{sec:conclusion}

In this work, we have presented a tractable data-driven solution to the Louisville problem of steering the distribution of a deterministic linear system from one normal distribution to another.
The problem was solved by decoupling the dynamics into mean and covariance dynamics and then performing an indirect data-driven design for the mean motion, and a direct data-driven design for the covariance motion through a convex relaxation.
The data-driven solution matches almost exactly with its model-based counterpart.
Further work in this direction will look at how to incorporate noise into the system dynamics.
This problem is much harder to solve, as the introduction of random noise (a) cannot be measured in a real-time data-driven scenario, and (b) yields input-state data that can have multiple realizations.

\section{Acknowledgment}

This work has been supported by NASA University Leadership Initiative award 80NSSC20M0163 and ONR award N00014-18-1-2828.
The article solely reflects the opinions and conclusions of its authors and not any NASA entity.

% % % % % % % % % % % % % % % % % % % % 

\bibliography{Refs_main}
\bibliographystyle{IEEEtran}

% % % % % % % % % % % % % % % % % % % % 

\appendix

\section*{A.~Proof of Theorem~\ref{theorem:1}}

\setcounter{equation}{0}
\renewcommand{\theequation}{A.\arabic{equation}}
\begin{proof}
	Given the matrix $W_0\in\Re^{(m + n)\times T}$, the Rouché-Capelli theorem \cite{RoucheCapelli} states that, for any given $q\in\Re^{m + n}$, the linear system $q = W_0 g$ admits an infinite number of solutions $g\in\Re^{T}$, given by
	\begin{equation}
		g = W_0^\dagger q + \Pi_{W_0} w, \quad w\in\Re^{T},
	\end{equation}
	where $\Pi_{W_0} := I_T - W_0^\dagger W_0$ is the orthogonal projector onto the nullspace of $W_0$.
	In the context of the mean dynamics \eqref{eq:meanDynamics}, letting $q = [v_k^\intercal \ \mu_k^\intercal]^\intercal$ and noting that $[B \ \ A]W_0 = X_{1,T}$ yields
	\begin{subequations}
		\begin{align}
			\mu_{k+1} &= [B \ \ A] W_0 g_k \\
			&= X_{1,T} \left(W_0^\dagger 
			\begin{bmatrix}
				v_k \\
				\mu_k
			\end{bmatrix} + \Pi_{W_0} w_k\right).
		\end{align}  
	\end{subequations}
	Theorem~\ref{theorem:1} now follows from $X_{1,T}\Pi_{W_0} = 0$.
\end{proof}

\section*{B.~Derivation of Relaxed Model-Based Covariance Steering}

\setcounter{equation}{0}
\renewcommand{\theequation}{B.\arabic{equation}}
For clarity, Problem~\ref{problem:2b} is given by
\begin{subequations}~\label{eq:covProblem}
	\begin{align}
		&\min_{\Sigma_k, K_k} J_{\Sigma} = \sum_{k=0}^{N-1}\Big(\mathrm{tr}(Q_k\Sigma_k) + \mathrm{tr}(R_k K_k \Sigma_k K_k^\intercal)\Big), \label{eq:covProblem_cost} \\
		&\textrm{such that, for all} \ k = 0,\ldots, N - 1, \nonumber \\
		&\Sigma_{k+1} = (A + B K_k)\Sigma_k (A + B K_k), \label{eq:covProblem_dyn} \\
		&\Sigma_0 = \Sigma_i, \quad \Sigma_N = \Sigma_f \label{eq:covProblem_BCs}.
	\end{align}
\end{subequations}
This is a nonlinear program in the decision variables $\Sigma_k$ and $K_k$. 
To remedy this, we first introduce the change of variables $U_k := K_k \Sigma_k$ \cite{CS_change_of_vars}, from which \eqref{eq:covProblem} can be written in the equivalent form
\begin{subequations}~\label{eq:covProblem2}
	\begin{align}
		&\min_{\Sigma_k, U_k} J_{\Sigma} = \sum_{k=0}^{N-1}\Big(\mathrm{tr}(Q_k\Sigma_k) + \mathrm{tr}(R_k U_k \Sigma_k^{-1} U_k^\intercal)\Big), \label{eq:covProblem2_cost} \\
		&\textrm{such that, for all} \ k = 0,\ldots, N - 1, \nonumber \\
		&\Sigma_{k+1} = A \Sigma_k A^\intercal + B U_k A^\intercal + A U_k^\intercal B^\intercal + B U_k \Sigma_k^{-1}U_k^\intercal B^\intercal, \label{eq:covProblem2_dyn} \\
		&\Sigma_0 = \Sigma_i, \quad \Sigma_N = \Sigma_f \label{eq:covProblem2_BCs}.
	\end{align}
\end{subequations}
The new optimization problem \eqref{eq:covProblem2} is still nonlinear in the decision variables $\Sigma_k, U_k$, but can be turned into a SDP through a convex relaxation as follows.
Define $Y_k \succeq U_k \Sigma_k^{-1} U_k^\intercal$, and relax the covariance dynamics to $\Sigma_{k+1} \succeq  A \Sigma_k A^\intercal + B U_k A^\intercal + A U_k^\intercal B^\intercal + B U_k \Sigma_k^{-1}U_k^\intercal B^\intercal$,
which leads to the relaxed problem
\begin{subequations}~\label{eq:covProblem3}
	\begin{align}
		&\min_{\Sigma_k, K_k} J_{\Sigma} = \sum_{k=0}^{N-1}\left(\mathrm{tr}(Q_k\Sigma_k) + \mathrm{tr}(R_k Y_k)\right), \label{eq:covProblem3_cost} \\
		&\textrm{such that, for all} \ k = 0,\ldots, N - 1, \nonumber \\
		&U_k \Sigma_k^{-1}U_k^\intercal - Y_k \preceq 0, \label{eq:covProblem3_ineq1} \\
		&A \Sigma_k A^\intercal + B U_k A^\intercal + A U_k^\intercal B^\intercal + B U_k \Sigma_k^{-1}U_k^\intercal B^\intercal - \Sigma_{k+1} \preceq 0, \label{eq:covProblem3_ineq2} \\
		&\Sigma_0 = \Sigma_i, \quad \Sigma_N = \Sigma_f \label{eq:covProblem3_BCs}.
	\end{align}
\end{subequations}
Using the Schur complement, constraints \eqref{eq:covProblem3_ineq1} and \eqref{eq:covProblem3_ineq2} can be written as the LMIs
\begin{subequations}~\label{eq:LMIs}
	\begin{align}
		\begin{bmatrix}
			\Sigma_k & U_k^\intercal \\
			U_k & Y_k
		\end{bmatrix} &\succeq 0, \label{eq:LMI1} \\
		\begin{bmatrix}
			\Sigma_k & U_k^\intercal B^\intercal \\
			B U_k & \Gamma_k
		\end{bmatrix} &\succeq 0, \label{eq:LMI2}
	\end{align}
\end{subequations}
where $\Gamma_k := \Sigma_{k+1} - A\Sigma_kA^\intercal - AU_k^\intercal B^\intercal - BU_kA^\intercal$.
This is equivalent to the top-level optimization problem in \eqref{eq:biLevelProblem}.

\end{document}